\documentclass[11pt]{article}
\usepackage{geometry}
\usepackage[T1]{fontenc}
\usepackage[utf8]{luainputenc}
\usepackage{amsmath}
\usepackage{amsthm}
\usepackage[english]{babel}

\usepackage{times}
\usepackage{amssymb}
\usepackage{cite}

\newtheorem{theorem}{Theorem}[section]

\newtheorem{lemma}[theorem]{Lemma}
\newtheorem*{lemma*}{Lemma}

\newtheorem{claim}[theorem]{Claim}
\newtheorem*{claim*}{Claim}
\newtheorem{definition}[theorem]{Definition}

\newtheorem{remark*}{Remark}

\DeclareMathOperator{\argmax}{arg\,max}
\newcommand{\ALG}{USM}
\newcommand{\LMS}{LMS_\epsilon}

\usepackage{pdfpages}

\begin{document}
\title{A Deterministic Algorithm for Maximizing Submodular Functions}
\author{Shahar Dobzinski \and Ami Mor}
\maketitle
\begin{abstract}
The problem of maximizing a non-negative submodular function was introduced by Feige, Mirrokni, and Vondrak [FOCS'07] who provided a deterministic local-search based algorithm that guarantees an approximation ratio of $\frac 1 3$, as well as a randomized $\frac 2 5$-approximation algorithm. An extensive line of research followed and various algorithms with improving approximation ratios were developed, all of them are randomized. Finally, Buchbinder et al. [FOCS'12] presented a randomized $\frac 1 2$-approximation algorithm, which is the best possible. 

This paper gives the first deterministic algorithm for maximizing a non-negative submodular function that achieves an approximation ratio better than $\frac 1 3$. The approximation ratio of our algorithm is $\frac 2 5$. Our algorithm is based on recursive composition of solutions obtained by the local search algorithm of Feige et al. We show that the $\frac 2 5$ approximation ratio can be guaranteed when the recursion depth is $2$, and leave open the question of whether the approximation ratio improves as the recursion depth increases.
\end{abstract}

\section{Introduction}

Let $f:2^M\rightarrow \mathbb R$ be a set function on a ground set $M$ of elements ($|M|=m$). We say that $f$ is \emph{submodular} if it exhibits decreasing marginal utilities. In other words, For every two sets $S$ and $T$ such that $S\subseteq  T\subseteq M$ and element $j\in M$: 
\begin{align*}
f(S\cup \{j\})-f(S)\geq f(T\cup \{j\})-f(T)
\end{align*}
An equivalent definition is that for every two subsets $S$ and $T$ it holds that:
$$
f(S)+f(T)\geq f(S\cup T)+f(S\cap T)
$$
It is common to assume that $f$ is non-negative: that is, for every $S\subseteq M$, $f(S)\geq 0$.

Submodular functions arise in many combinatorial optimization scenarios. For example, given a graph $G=(V,E)$, let the cut function of $G:2^V\rightarrow \mathbb R$ be the function that assigns every subset of the vertices $S$ the size of the cut $(S,V-S)$. Observe that a cut function is submodular. Similarly, the objective function in classic problems such as SET-COVER is also submodular. The decreasing marginal utilities property also makes submodular functions a popular object of study in various economic settings, such as combinatorial auctions and influence maximization in social networks (see, e.g., \cite{Lehmann:2001:CAD:501158.501161} and \cite{kempe2003maximizing}, respectively). 

The formal study of approximation algorithms\footnote{In general representing a submodular function $f$ takes space that is exponential in $m$. Therefore, we assume that $f$ is accessed via value queries: given $S$, what is $f(S)$? All the algorithms we mention make $poly(m)$ value queries and run in time $poly(m)$.} for submodular optimization problems has started back in 1978: Nemhauser, Wolsey and Fisher \cite{nemhauser1978analysis}  showed that for monotone and non-negative submodular functions, a natural greedy algorithm gives a $(1-1/e)$-approximation to \emph{monotone} submodular maximization under a cardinality constraint.

However, the objective function of problems like $MAX-CUT$ is submodular but not necessarily monotone so the above algorithm does not guarantee any bounded approximation ratio. To handle this lacuna, Feige et al. \cite{feige2011maximizing} initiated the study of algorithms for maximizing non-negative submodular functions. They developed a local search algorithm that finds a set whose value is at least $\frac{1}{3}$ of the value of the maximum-value set. In addition, they presented a randomized $\frac{2}{5}$-approximation.
Their algorithms are complemented by the following hardness result: any algorithm that obtains a better approximation ratio than $\frac{1}{2}$, must make exponentially (in $m$) many queries.
This hardness result was later extended \cite{DBLP:journals/corr/abs-1202-2792}: unless $NP=RP$, no $\frac{1}{2}+\epsilon$ polynomial time approximation algorithm exists.

Subsequent work led to improvements in the approximation ratio. Oveis Gharan and Vondrak showed a randomized $0.41$-approximation algorithm \cite{DBLP:journals/corr/abs-1007-1632}.
Feldman et al. \cite{Feldman:2011:NSM:2027127.2027164} improved this to $0.42$, again using a randomized algorithm.
Finally, Buchbinder, Feldman, Naor and Schwartz \cite{Buchbinder:2012:TLT:2417500.2417835} presented a randomized algorithm that achieves an approximation ratio of $\frac 1 2$, which as mentioned above, is the best possible.

However, there was no progress at all regarding the approximation ratio achievable by deterministic algorithms, and prior to our work no deterministic algorithm was known to obtain a ratio better than $\frac{1}{3}$. Indeed, this open question was recently mentioned, e.g., by \cite{feldman2014maximizing}.
Our paper makes progress on this front:

\begin{theorem} \label{mainresults1}
For every $\epsilon>0$ there is a deterministic algorithm that provides a $(\frac 2 5-\epsilon)$-approximation to unconstrained submodular maximization and makes $O\left(\frac{1}{\epsilon} \cdot m \cdot \log m\right)$ value queries.
\end{theorem}
Our algorithm can be seen as a recursive composition of solutions obtained by the local search algorithm of \cite{feige2011maximizing}. Here is a sketch of our algorithm (see the technical sections for a more accurate description):
\begin{enumerate}
\item Find a local maximum\footnote{$S$ is a local maximum if $f(S)\geq f(S\cup T)$ and $f(S)\geq f(S-T)$ for all $T$.} $S$.
\item\label{alg-step2} Use an approximation algorithm to obtain $T_1$ with high value when $f$ is restricted to $S^c$. \footnote{We use the notation $S^c$ to denote the complement set of $S$.}
\item\label{alg-step3} Use an approximation algorithm to obtain $T_2\subseteq S$ such that $f\left(S^c\cup T_2\right)$ is high.
\item Return  $\argmax \left\{S,S^c,T_1\cup T_2\right\}$.
\end{enumerate}
For the base case we use in Steps \ref{alg-step2} and \ref{alg-step3} the local search algorithm of \cite{kempe2003maximizing}. After obtaining an algorithm with an improved approximation ratio, this new algorithm can be used recursively in Steps \ref{alg-step2} and \ref{alg-step3}. We will show that when the recursion depth is $2$ the approximation ratio is $\frac 2 5$. The approximation ratio of the algorithm should in principle improve as the recursion depth increases, but we do not know how to analyze it. This is the main open question that this paper leaves.

\section{A $\frac 2 5$-Approximation Algorithm}

In this section we prove that there exists a $\frac 2 5$-approximation algorithm to the problem of maximizing a non-negative submodular function. To gain some intuition about our algorithm we will see how one can get a better than $\frac 1 3$-approximation (Subsection \ref{subsec-intuition}). We will then give a formal description and analysis of our algorithm (Subsection \ref{subsec-algorithm} onwards).

\subsection{Warm-Up: Improving over the $\frac 1 3$-Approximation Factor}\label{subsec-intuition}

We now show that we can get an approximation ratio better than $\frac 1 3$ approximation. Let $S$ be a local maximum $S$. Feige et al. \cite{feige2011maximizing} show that $2f\left(S\right) + f\left(S^c\right)\geq f\left(OPT\right)$,
which implies that either $S$ or $S^c $ provides a $\frac{1}{3}$ approximation (we let OPT denote the optimal solution). In fact, a slightly stronger inequality holds: $2f\left(S\right) + f\left(S^c\right)\geq f\left(OPT\right)+f\left(M\right)+f\left(\emptyset\right)$.

Consider a scenario in which the local-maximum algorithm returns only a $\frac{1}{3}$ approximation.
In that case we have $f\left(S\right)= \frac{f(OPT)}{3}$ and $f\left(S^c\right)= \frac{f(OPT)}{3}$.
We will show that two other sets, $S^c\cup OPT$ and $S^c\cap OPT$ have high value. For the first set, observe that (using submodularity for the first inequality, $S$ being a local maximum for the second, and $f(S)= \frac {f(OPT)} 3$ for the third one):
\begin{align}\label{eqn-intuition}
f\left(S^c\cup OPT\right) &\geq f\left(OPT\right) + f\left(M\right) - f\left(OPT\cup S\right)\\\nonumber
&\geq f\left(OPT\right) + f\left(M\right) - f\left(S\right) \\\nonumber
&\geq ~\frac{2f(OPT)}{3}\nonumber
\end{align}
Similarly $f\left(S^c \cap OPT\right)\geq ~\frac{2f(OPT)}{3}$.
In other words:
\begin{enumerate}
\item There is a subset of $S^c$ with high value. In particular, the value is at least $2f(S^c)$.
\item There is a superset of $S^c$ with high value. In particular, the value is at least $2f(S^c)$.
\end{enumerate}
Now we would like to (approximately) find these subsets. Observe that we actually have a (smaller) submodular maximization problem. To see this, let $M_1 = S^c$ and $f_1:2^{M_1}\rightarrow \mathbb R$ be the submodular function $f_1(T)=f(T)$ for every $T\subseteq M_1$. Observe that $f(S^c)=f\left(M_1\right) = \frac{f(OPT)}{3}$. In addition, we have already seen that $f(OPT_{M_1})\geq \frac {f(OPT)} 3 $, where $OPT_{M_1}$ is a maximum value set of $f_1$.

We now run the local search algorithm on the the function $f_1$ and find a local maximum $T_1\subseteq S^c$. By the properties of a local maximum we are guaranteed that
$$
2f\left(T_1\right) + f\left(M_1-T_1\right)\geq f\left(OPT_{M_1}\right)+f\left(M_1\right)+f\left(\emptyset\right)
$$
Notice that by our bounds for $f(OPT_{M_1})$ and $f(M_1)$ it must be that $f(T_1)\geq \frac {f(OPT)} 3 $ or that $f(M_1-T_1)\geq \frac {f(OPT)} 3 $. If one of them has higher value we already reach an approximation ratio better than $\frac 1 3$. Hence, we assume that $f\left(T_1\right) = \frac{f(OPT)}{3}$, then $f\left(M_1-T_1\right)=\frac{f(OPT)}{3}$.
 
We now repeat a similar argument. Similarly to inequality (\ref{eqn-intuition}),
\begin{align}\label{eqn-f'_1}
f\left((M_1-T_1)\cup OPT_{M_1}\right) \geq f\left(OPT_{M_1}\right) + f\left(M_1\right) - f\left(T_1\right)\geq \frac{f(OPT)}{3}
\end{align}
Let $M'_1=T_1$. Let $f'_1:2^{M'_1}\rightarrow\mathbb R$ be the submodular function where $f'_1(T)=f(T\cup (M_1-T_1)$ for every $T\subseteq M'_1$. Observe that $f'_1\left(M'_1\right)= \frac{f(OPT)}{3} $ and $f'_1\left(\emptyset\right)= \frac{f(OPT)}{3}$. Furthermore, by inequality (\ref{eqn-f'_1}) the maximum value of $f'_1$ is at least $\frac {2f(OPT)} 3 $.
Let $T'_1$ be a local maximum of $f'_1$. Using \cite{feige2011maximizing} again we have that: 
$$
\max(f(T'_1),f(M'_1-T'_1))\geq \frac{1}{3} \left(f\left(M_1-T_1\right)+f\left(M_1\right)+f\left(OPT_{M_1}\cup (M_1-T_1)\right) \right)\geq \frac{4}{9}\cdot f(OPT)
$$
This already establishes that the approximation ratio is better than $\frac 1 3$, but we can do even better: find $T_2\subseteq S$ such that $f\left(S^c \cup T_2\right)$ has high value.
Combine $T_1$ and $T_2$ to sum the improvements, since submodularity implies that
$f\left(T_1\cup T_2\right)\geq f\left(S^c\right) + \left(f\left(T_1\right)-f\left(S^c\right)\right) + \left(f\left(S^c \cup T_2\right)-f\left(S^c\right)\right)$.

\subsection{The Algorithm} \label{subsec-algorithm}

We now present our main result: an improved deterministic algorithm for maximizing a submodular function. As discussed earlier, our algorithm will be based on repeatedly finding local maxima of different subsets and combining the result. However, we do not know how to find a local maximum with a polynomial number of value queries, thus we have to settle on approximate solutions:

\begin{definition}
Fix a function $f:2^M\rightarrow \mathbb{R}$. $S$ is a \emph{{}$\left(1+\epsilon\right)-\text{local maximum}$} of $f$ if for every $T\subset M$:
\begin{enumerate}
\item $\left(1+\epsilon\right)f\left(S\right) \geq f\left(S\cup T\right)$.
\item $\left(1+\epsilon\right)f\left(S\right) \geq f\left(S\cap T\right)$.
\end{enumerate}
\end{definition}

We will use algorithms from \cite{feige2011maximizing} and \cite{Buchbinder:2012:TLT:2417500.2417835} for our
``base case''.
\begin{theorem}[essentially \cite{feige2011maximizing, Buchbinder:2012:TLT:2417500.2417835}]\label{theoremm}
There is an algorithm $\LMS$ that gets as a input function $f$, a set of elements $M$ and $\epsilon$ and returns a set $S$ that is $\left(1+\epsilon\right)-\text{approximate local maximum}$. In addition, $f(S)\geq \frac {f(OPT)} 3$, where OPT is the set with the highest value. The algorithm makes $O\left(\frac{1}{\epsilon}\cdot m^2\cdot \log m\right)$ value queries.
\end{theorem}
Feige et al \cite{feige2011maximizing} show how to start with a set $T$ and find a set $S$ that is a $(1+\epsilon)$-approximate local maximum with $O\left(\frac{1}{\epsilon}\cdot m^2\cdot \log m\right)$ value queries. Their algorithm guarantees that $f(S)\geq f(T)$ and that either $f(S)\geq \frac {f(OPT)} {3(1+\epsilon)}$ or $f(S^c)\geq \frac {f(OPT)} {3(1+\epsilon)}$. This weaker guarantee suffices to our analysis, but it would be simpler to first obtain a set $T$ with $f(T)\geq \frac {f(OPT)} {3}$ via the deterministic algorithm of \cite{Buchbinder:2012:TLT:2417500.2417835}, and then obtain an approximate local maximum $S$ by running the algorithm of \cite{feige2011maximizing} with $T$ as the initial set.

We are now ready to present our recursive algorithm for unconstrained submodular maximization. The input is a submodular function $f$ defined on a set of elements $M$, and a parameter $nrounds$ that determines the recursion depth. We also assume some fixed accuracy parameter $\epsilon>0$. The algorithm returns a set that will be proved to provide a high value. 

\subsubsection*{ $\ALG_{\epsilon}\left(f,M,nrounds\right):$}
\begin{enumerate}
\item Define $f'$: $f'\left(T\right)=f\left(T\right)-\min\left(f\left(\emptyset\right),f\left(M\right)\right)$. Let $S=\LMS\left(f',M\right)$.
\item\label{step-stopping} If $S=M$, $S=\emptyset$, or $nrounds=0$ then return $S$. 
\item\label{T1}  Let $M_{1}=S^{c}$, define $f_1$: $f_1\left(T\right)=f'\left(T\right)$. Let $T_{1}=\ALG_{\epsilon}\left(f_{1},M_{1},nrounds-1\right)$.
\item\label{T2} Let $M_{2}=S$, define $f_2$: $f_{2}\left(T\right)=f'\left(S^{c}\cup T\right)$. Let $T_{2}=\ALG_{\epsilon}\left(f_{2},M_{2},nrounds-1\right)$.
\item Return $\argmax_{R\in \left\{S,\left(T_{1}\cup T_{2}\right),M,\emptyset\right\}}f\left(R\right)$.
\end{enumerate}

\begin{theorem}\label{thm-main}
Let $f(\ALG_{\epsilon}\left(f,M,2\right))\geq (\frac {2} 5-\epsilon)\cdot f(OPT)$. Moreover, the algorithm makes $O\left(\frac{1}{\epsilon}\cdot m^2\cdot \log m\right)$ value queries in this case.
\end{theorem}

\subsection{Proof of Theorem \ref{thm-main}}

We first analyze the running time of the algorithm. Let $L_{\epsilon}\left(m\right)$ be the maximum number of queries $\LMS$ makes on any submodular function that is defined on $m$ elements.

\begin{lemma}
For every integer $n\geq 0$ the number of queries used by $\ALG_{\epsilon}\left(f,M,n\right)$ is $O\left(m\cdot L_{\epsilon}\left(m\right)\right)$. As a corollary, for every integer $n\geq 0$ the number of queries used by $\ALG_{\epsilon}\left(f,M,n\right)$ is $O\left(\frac{1}{\epsilon} \cdot m^3\log m \right)$.
\end{lemma}
\begin{proof}
Denote by $T_\epsilon\left(k\right)$ the maximal number of queries $\ALG_{\epsilon}$ makes on a submodular function with $|M|=k$.
In iterations where the algorithm stops at Step \ref{step-stopping}, the number of queries is at most $L_{\epsilon}\left(|M|\right)$.
Else, we have that $|M|>|S|>0$. The algorithm then recursively solves two subproblems,
one on $S$ and the other on $S^c$.
Thus, $T_\epsilon\left(m\right)$ can be bounded by:
$$T_\epsilon\left(m\right)\leq \max_{d=1..m-1}\left(T_\epsilon\left(d\right)+T_\epsilon\left(m-d\right)\right) + L_{\epsilon}\left(m\right)$$
Solving that recursion by induction on $m$ gives that $T_\epsilon\left(m\right) \leq L_{\epsilon}\left(m\right) \cdot m$.
\end{proof}
The next lemma bounds the approximation ratio when the recursion depth is $2$ (i.e., $nrounds=2$). We note that using a larger value of $nrounds$ should in principle yield an improved ratio, but currently we are unable to formally analyze this case. The next subsections are devoted to proving the lemma. 

\begin{lemma} \label{BIGLEMMA}
For every $\epsilon>0$ and for every submodular function $f$ with $f\left(M\right)\geq 0$ and $f\left(\emptyset\right)\geq 0$:
$$f\left(\ALG_{\epsilon}\left(f,M,2\right)\right)\geq\left(\frac{2}{5}- \epsilon\right) f\left(OPT\right) $$
\end{lemma}

\subsection{Proof Overview}

We introduce a function $\alpha_{i}$ which represents a lower bound on the value of the solution that the algorithm outputs after $i$ rounds. Let $\mathcal F$ be the set of all submodular functions.
\begin{definition}
$\alpha_{i}(x_{OPT},x_{0},x_{M}) = \inf_{f\in \mathcal F} {\left\{\ALG_{\epsilon}\left(f,M,i\right) | f\left(OPT\right)\geq x_{OPT}, f\left(\emptyset\right)\geq x_{0}, f\left(M\right) \geq x_{M} \right\}}$.
\end{definition}
Our result is based on the following lower bounds for $\alpha_{i}$. The proofs are in Subsection \ref{subsec-proofs}.
\begin{claim} \label{ROUND0}
For every non negative $x_{OPT},x_0,x_M$:
$$\alpha_{0}(x_{OPT},x_{0},x_{M}) \geq \max(\frac{1}{3}\left(x_{OPT}+x_{0}+x_{M}\right),x_{M},x_{0})$$
\end{claim}
\begin{claim}[composition lemma] \label{GENERAL} 
Let $f$ be a submodular function, $i>0$, and $S$, $T_1$ and $T_2$ are the sets defined in the algorithm when we run $\ALG_\epsilon(f,M,i)$.
Then,
\begin{eqnarray*}
f\left(T_{1}\cup T_{2}\right) & \geq & \alpha_{i-1}\left(f\left(OPT\right)+f\left(\emptyset\right)-\left(1+\epsilon\right)f\left(S\right),f\left(\emptyset\right),f\left(S^{c}\right)\right) +\\
&& \alpha_{i-1}\left(f\left(OPT\right)+f\left(M\right)-\left(1+\epsilon\right)f\left(S\right),f\left(S^{c}\right),f\left(M\right)\right) -f\left(S^{c}\right)
\end{eqnarray*}
And in particular:
$$ f\left(T_{1}\cup T_{2}\right) \geq \alpha_{i-1}\left(f\left(OPT\right)+f\left(M\right)-\left(1+\epsilon\right)f\left(S\right),f\left(S^{c}\right),f\left(M\right)\right)$$

\end{claim}
\begin{claim} \label{ROUND1}
For every non negative $x_{OPT},x_M,x_0$ we have:
\begin{eqnarray*}
\alpha_{1}(x_{OPT},0,x_{M}) &\geq &\max(\frac{1-\epsilon}{3}\cdot x_{OPT}+\frac{1-\epsilon}{2}\cdot x_{M},x_{0}) \\
\alpha_{1}(x_{OPT},x_0,0) &\geq &\max(\frac{1-\epsilon}{3}\cdot x_{OPT}+\frac{1-\epsilon}{2}\cdot x_{0},x_{0}) 
\end{eqnarray*}
\end{claim}
Before proving these two claims, let us see why they let us derive our main bound on the quality of the solution (note that lemma \ref{BIGLEMMA} can be directly derived from this claim):
\begin{claim} \label{ROUND2}
For every non negative $x_{OPT}$, $x_M$ we have:
$$\alpha_{2}(x_{OPT},0,0) \geq \left(\frac{2}{5}- \epsilon\right)x_{OPT}$$
\end{claim}
\begin{proof} 
Let $f$ be a submodular function such that $f\left(OPT\right) \geq x_{OPT}$, $f\left(\emptyset\right) \geq 0$ and $f\left(M\right) \geq 0$.
We want to show that:
\begin{equation*} 
\ALG_{\epsilon}\left(f,M,2\right)\geq \left(\frac{2}{5}- \epsilon\right)f\left(OPT\right)
\end{equation*}
Let $T_1$, $T_2$ be the sets defined by the algorithm when we run $\ALG_{\epsilon}\left(f,M,2\right)$.
Using the composition lemma we know that:
\begin{eqnarray*}
f\left(T_1 \cup T_2 \right) & \geq & \alpha_{1}\left(f\left(OPT\right)-\left(1+\epsilon\right)f\left(S\right),0,f\left(S^{c}\right)\right) +                                    \alpha_{1}\left(f\left(OPT\right)-\left(1+\epsilon\right)f\left(S\right),f\left(S^{c}\right),0\right) -f\left(S^{c}\right)
\end{eqnarray*}
We use claim \ref{ROUND1} to bound $\alpha_1$ and get:
\begin{eqnarray*}
f\left(T_1 \cup T_2 \right) &\geq & 2\left(1-\epsilon\right)\left(\frac{1}{3}\left[f\left(OPT\right)-\left(1+\epsilon\right)f\left(S\right)\right]+\frac{1}{2}f\left(S^{c}\right)\right)-f\left(S^c\right) \\
& = & \frac{2}{3}\left(1-\epsilon\right)f\left(OPT\right)-\frac{2}{3}\left(1-\epsilon^2\right)f\left(S\right)-\epsilon f\left(S^c\right) \\
& \geq &\frac{2}{3}\left(1-\epsilon\right)f\left(OPT\right)-\frac{2}{3}f\left(S\right)-\epsilon f\left(OPT\right) \\
& = &\frac{2}{3}\left(f\left(OPT\right)-f\left(S\right)\right) - \frac{5}{3}\epsilon f\left(OPT\right)
\end{eqnarray*}
The algorithm compares $f\left(T_1 \cup T_2 \right)$ to $f\left(S\right)$ and returns the set with highest value.
Therefore,
$$\ALG_{\epsilon}\left(f,M,2\right)\geq \max\left(f\left(S\right),\frac{2}{3}f\left(OPT\right)-\frac{2}{3}f\left(S\right) - \frac{5}{3}\epsilon f\left(OPT\right)\right)$$
That is a maximum of two functions, both linear in $f\left(S\right)$, one is ascending and one is descending. The minimum of a 
maximum of two such functions is achieved at their intersection. I.e., when $f\left(S\right)=\frac{2}{3}\cdot f\left(OPT\right)-\frac{2}{3}\cdot f\left(S\right)-\frac{5}{3}\cdot \epsilon \cdot f\left(OPT\right)$.
Solving this equation gives us:
$$\ALG_{\epsilon}\left(f,M,2\right) \geq \left(\frac{2}{5}- \epsilon\right)f\left(OPT\right)$$
By the definition of $\alpha_2$ and since $f\left(OPT\right)\geq x_{OPT}$,
$$\alpha_{2}\left(x_{OPT},0,0\right)\geq \left(\frac{2}{5}- \epsilon\right)x_{OPT}$$

\end{proof}

\subsection{Proofs}\label{subsec-proofs}
To finish the proof of our main result all we are left with is proving the bounds for
$\alpha_0$ and $\alpha_1$ (claims \ref{ROUND0}, \ref{ROUND1}) and the composition lemma (claim \ref{GENERAL}).
We start with some helpful observations.
\subsubsection{Observations}

Let $S$ be a $(1+\epsilon)$-local maximum. Therefore:

\begin{equation} \label{LMCAP}
\left(1+\epsilon\right)f\left(S\right)\geq f\left(S\cap OPT\right)
\end{equation}
\begin{equation} \label{LMCUP}
\left(1+\epsilon\right)f\left(S\right)\geq f\left(S\cup OPT\right)
\end{equation}
Also, submodularity of $f$ implies that
\begin{equation}  \label{SUBMODCAP}
f\left(OPT\cap S^{c}\right)+f\left(OPT\cap S\right)\geq f\left(OPT\right)+f\left(\emptyset\right)
\end{equation}
\begin{equation} \label{SUBMODCUP}
f\left(OPT\cup S^{c}\right)+f\left(OPT\cup S\right)\geq f\left(OPT\right)+f\left(M\right)
\end{equation}
The next claim appeared in \cite{feige2011maximizing} in a slightly different form:
\begin{claim} \label{LMSTHIRD}
Let $S$ be a $(1+\epsilon)$-local maximum, and $f$ a submodular function. Then $2\left(1+\epsilon\right)f\left(S\right)+f\left(S^{c}\right) \geq f\left(OPT\right)+f\left(M\right)+f\left(\emptyset\right)$.
\end{claim}
\begin{proof}
Since $f$ is submodular, 
\begin{equation*}
f\left(S^{c}\right)+f\left(OPT\cup S\right)\geq f\left(OPT\cap S^c\right)+f\left(M\right).
\end{equation*}
Inequalities  \eqref{LMCAP}, \eqref{LMCUP}, \eqref{SUBMODCAP} together with the above inequality gives us:

\begin{eqnarray*} \label{FEIGE}
2\left(1+\epsilon\right)f\left(S\right)+f\left(S^{c}\right) & \geq & f\left(S\cap OPT\right) + f\left(S\cup OPT\right) + f\left(S^{c}\right) \\
 & \geq & f\left(S\cup OPT\right)+f\left(OPT\cap S^{c}\right)  +f\left(M\right)\\
 & \geq & f\left(OPT\right)+f\left(\emptyset\right)+f\left(M\right)
\end{eqnarray*}

\end{proof}

\subsubsection{Proof of Claim \ref{ROUND0}}

Let $f$ be a submodular function, and $S$ be the local maximum found by the algorithm.
Using theorem \ref{theoremm}:
$$3f\left(S\right) \geq f\left(OPT\right)+f\left(M\right)+f\left(\emptyset\right)$$

Also, note that the algorithm takes the set with highest value among $S$, $S^c$, $M$ and $\emptyset$. Therefore, for every submodular $f$:
$$\ALG_{\epsilon}\left(f,M,0\right)\geq\max\left(f\left(\emptyset\right),f\left(M\right),\frac{1}{3}\left(f\left(OPT\right)+f\left(M\right)+f\left(\emptyset\right)\right)\right)$$
Hence, by the definition of $\alpha_0$, 
$$\alpha_0\geq\max\left(x_{o},x_{M},\frac{1}{3}\left(x_{OPT}+x_{0}+x_{M}\right)\right)$$

\subsubsection{Proof of Claim \ref{GENERAL} (composition lemma)}

We follow steps \ref{T1} and \ref {T2} of the algorithm. The algorithm takes the set $S^{c}$
and attempts to improve it by adding some elements and removing others. 
Since $f$ is submodular and $T_1\subseteq S^c$, then $f\left(T_{2}|T_{1}\right)\geq f\left(T_{2}|S^{c}\right)$.
Therefore:
\begin{equation} \label{submodt1t2}
f\left(T_{1}\cup T_{2}\right)=f\left(T_{1}\right)+f\left(T_{2}|T_{1}\right)\geq f\left(T_{1}\right)+f\left(T_{2}|S^{c}\right)=f_{1}\left(T_{1}\right)+f_{2}\left(T_{2}\right)-f\left(S^{c}\right)
\end{equation}
Next, we give bounds on $f_{1}\left(T_{1}\right)$ and $f_{2}\left(T_{2}\right)$:
\begin{claim} \label {SUBESETHELPER}
The above $T_1$ and $T_2$ satisfy:
\begin{eqnarray*}
f_{1}\left(T_{1}\right)&\geq & \alpha_{i-1}\left(f\left(OPT\right)+f\left(\emptyset\right)-f\left(S\right),f\left(\emptyset\right),f\left(S^{c}\right)\right)\\
f_{2}\left(T_{2}\right)&\geq & \alpha_{i-1}\left((f\left(OPT\right)+f\left(M\right)-f\left(S\right),f\left(S^{c}\right),f\left(M\right)\right)
\end{eqnarray*}
\end{claim}

\begin{proof}
First, note that by the definition of $\alpha_i$, it is monotone in the following sense: 
if $x_{OPT}\geq y_{OPT}$, $x_0 \geq y_0$ and $x_M \geq y_M$, then 
$\alpha_i\left(x_{OPT},x_0,x_M\right) \geq \alpha_i\left(y_{OPT},y_0,y_M\right)$.

Consider the subproblems in which we attempt to find a maximum of $f_1$ and $f_2$.
By the definition of $\alpha_{i-1}$, we have:
\begin{equation} \label{f1}
f_1\left(T_1\right)=\ALG_{\epsilon}\left(f_1,M_1,i-1\right) \geq \alpha_{i-1}\left(\max_{R\subseteq S^c}\left(f_1\left(R\right)\right),f_1\left(\emptyset\right),f_1\left(M_1\right)\right)
\end{equation}
\begin{equation} \label{f2}
f_2\left(T_2\right)=\ALG_{\epsilon}\left(f_2,M_2,i-1\right) \geq \alpha_{i-1}\left(\max_{R\subseteq S}\left(f_2\left(R\right)\right),f_2\left(\emptyset\right),f_2\left(M_2\right)\right)
\end{equation}

Combining \eqref{LMCAP} and \eqref{SUBMODCAP} gives us
$f_{1}\left(OPT\cap S^c\right)\geq f\left(OPT\right)+f\left(\emptyset\right)-\left(1+\epsilon\right)f\left(S\right)$. 
Therefore, $\max_{R\subseteq S^c}f_1\left(R\right)\geq f_{1}\left(OPT\cap S^{c}\right)\geq f\left(OPT\right)+f\left(\emptyset\right)-\left(1+\epsilon\right)f\left(S\right)$.
Note also that $f_{1}\left(\emptyset\right)=f\left(\emptyset\right)$ and
$f_{1}\left(M_{1}\right)=f\left(S^{c}\right)$.
Applying inequality \eqref{f1}  and using the monotonicity of $\alpha_i$, we have: 
\begin{eqnarray*}
f_{1}\left(T_{1}\right)&\geq & \alpha_{i-1}\left(\max_{R\subseteq S^c}\left(f_1\left(R\right)\right),f_1\left(\emptyset\right),f_1\left(M_1\right)\right)\\
&\geq& \alpha_{i-1}\left(f\left(OPT\right)+f\left(\emptyset\right)-\left(1+\epsilon\right)f\left(S\right),f\left(\emptyset\right),f\left(S^{c}\right)\right)
\end{eqnarray*}

Note that $f_{2}\left(OPT\cap S\right)=f\left(OPT\cup S^{c}\right)$.
Combining inequalities \eqref{LMCUP} and \eqref{SUBMODCUP} together we get
$f_{2}\left(OPT\cap S\right)\geq f\left(OPT\right)+f\left(M\right)-\left(1+\epsilon\right)f\left(S\right)$. 
Therefore, $\max_{R\subseteq S}f_2\left(R\right)\geq f_2\left(OPT \cap S\right)\geq f\left(OPT\right)+f\left(M\right)-\left(1+\epsilon\right)f\left(S\right)$.
Note also that $f_{2}\left(\emptyset\right)=f\left(S^c\right)$ and
$f_{2}\left(M_{2}\right)=f\left(M\right)$.
Applying inequality \eqref{f2}  and using the monotonicity of $\alpha_i$, we have: 
\begin{eqnarray*}
f_{2}\left(T_{2}\right)&\geq & \alpha_{i-1}\left(\max_{R\subseteq S}f_2\left(R\right),f_2\left(\emptyset\right),f_1\left(M_2\right)\right)\\
& \geq & \alpha_{i-1}\left(f\left(OPT\right)+f\left(M\right)-\left(1+\epsilon\right)f\left(S\right),f\left(S^{c}\right),f\left(M\right)\right)
\end{eqnarray*}

\end{proof}

Back to the proof of claim \ref{GENERAL}, by \eqref{submodt1t2} we have $f\left(T_{1}\cup T_{2}\right)\geq f\left(T_1\right)+f\left(T_2\right)-f\left(S^c\right)$. 
Applying claim \ref{SUBESETHELPER} gives us:
\begin{eqnarray*}
f\left(T_{1}\cup T_{2}\right) & \geq & \alpha_{i-1}\left(f\left(OPT\right)+f\left(\emptyset\right)-\left(1+\epsilon\right)f\left(S\right),f\left(\emptyset\right),f\left(S^{c}\right)\right) +\\ 
&& \alpha_{i-1}\left(f\left(OPT\right)+f\left(M\right)-\left(1+\epsilon\right)f\left(S\right),f\left(S^{c}\right),f\left(M\right)\right) -f\left(S^{c}\right)
\end{eqnarray*}
To prove second part of the composition lemma, we use the following claim:
\begin{claim} \label{trivialM}
For every non-negative $x_{OPT},x_M,x_0$ and every $i$ we have:
$\alpha_i\left(x_{OPT},x_0,x_M\right)\geq x_M$.
\end{claim}
\begin{proof}
In Step \ref{step-stopping} the algorithm can choose to output $M$. Therefore, $\ALG_{\epsilon}\left(f,M,i\right)\geq f\left(M\right)$.
\end{proof}
Using claim \ref{trivialM} we know that 
$$\alpha_{i-1}\left(f\left(OPT\right)+f\left(\emptyset\right)-\left(1+\epsilon\right)f\left(S\right),0,f\left(S^{c}\right)\right) \geq f\left(S^c\right)$$
Combining it with the first part of the composition lemma gives us:
\begin{eqnarray*}
f\left(T_{1}\cup T_{2}\right) & \geq & \alpha_{i-1}\left(f\left(OPT\right)+f\left(\emptyset\right)-\left(1+\epsilon\right)f\left(S\right),f\left(\emptyset\right),f\left(S^{c}\right)\right) +\\ 
&& \alpha_{i-1}\left(f\left(OPT\right)+f\left(M\right)-\left(1+\epsilon\right)f\left(S\right),f\left(S^{c}\right),f\left(M\right)\right) -f\left(S^{c}\right)\\ 
&\geq & \alpha_{i-1}\left(f\left(OPT\right)+f\left(M\right)-\left(1+\epsilon\right)f\left(S\right),f\left(S^{c}\right),f\left(M\right)\right)
\end{eqnarray*}

\subsubsection {Proof of Claim \ref{ROUND1}}

Let $f$ be a submodular function such that $f\left(OPT\right)\geq x_{OPT}$, $f\left(\emptyset\right)\geq 0$ and $f\left(M\right)\geq x_M$.
We will show that:
\begin{equation*}  
\ALG_{\epsilon}\left(f,M,1\right)\geq \max\left(x_{M},\frac{1}{3}\left(f\left(OPT\right)+f\left(M\right)\right)+\frac{1}{6}f\left(M\right)\right)
\end{equation*}
Let $T_1$, $T_2$ be the sets defined by the algorithm when we run $\ALG_{\epsilon}\left(f,M,1\right)$. Using the second part of the composition lemma we know that:
\begin{equation} \label{LHS1}
f\left(T_{1}\cup T_{2}\right) \geq \alpha_{0}\left(f\left(OPT\right)+f\left(M\right)-\left(1+\epsilon\right)f\left(S\right),f\left(S^{c}\right),f\left(M\right)\right)
\end{equation}
Claim \ref{ROUND0} shows that $\alpha_{0}\left(x_{OPT},x_{0},x_{M}\right)\geq \frac{1}{3}\left(x_{OPT}+x_{0}+x_{M}\right)$, 
and claim \ref{LMSTHIRD} gives us: $f\left(S^c\right)\geq f\left(OPT\right)+f\left(M\right)-2\left(1+\epsilon\right)f\left(S\right)$.
Therefore,
\begin{align} \label{RHS12}
\alpha_{0}&\left(f\left(OPT\right)+f\left(M\right)-\left(1+\epsilon\right)f\left(S\right),f\left(S^{c}\right),f\left(M\right)\right) \nonumber \\ 
& \geq  \frac{\left(f\left(OPT\right) + 2f\left(M\right)-\left(1+\epsilon\right)f\left(S\right)+f\left(S^{c}\right)\right)}{3}\nonumber \\   
&\geq \frac{2f\left(OPT\right)+3f\left(M\right)-\left(1+\epsilon\right)3f\left(S\right)}{3} 
\end{align} 
From inequalities \eqref{LHS1} and \eqref{RHS12} it follows that:
$$f\left(T_1 \cup T_2\right) \geq \frac{2f\left(OPT\right)+3f\left(M\right)-\left(1+\epsilon\right)3f\left(S\right)}{3}$$
The value of the output of the algorithm is at least the maximum of $f\left(T_1\cup T_2\right)$, $f\left(S\right)$, $\emptyset$ and $M$. Hence,
$$\ALG_{\epsilon}\left(f,M,1\right)\geq \max\left\{f\left(S\right),\frac{2f\left(OPT\right)+3f\left(M\right)-\left(1+\epsilon\right)3f\left(S\right)}{3}\right\}$$
This is a maximum of two functions, both linear in $f\left(S\right)$, one is ascending and one is descending. The minimum of a 
maximum of two such functions is obtained when they intersect, i.e., when $f\left(S\right)=f\left(T_1\cup T_2\right)$.
Therefore
$\max\left(f\left(S\right),f\left(T_1 \cup T_2\right)\right) \geq \frac{1}{6+3\epsilon}\left(2f\left(OPT\right)+3f\left(M\right)\right)$.
Since $\frac{1}{6+3\epsilon}\geq \frac{1}{6}\left(1-\epsilon\right)$ we have:
$$\max\left(f\left(S\right),f\left(T_1 \cup T_2\right)\right) \geq \left(1-\epsilon\right)\left(\frac{1}{3}\cdot f\left(OPT\right)+\frac{1}{2}\cdot f\left(M\right)\right)$$
Therefore,
$$\ALG_{\epsilon}\left(f,M,1\right)\geq \max\left(f\left(M\right),\frac{1-\epsilon}{3}\left(f\left(OPT\right)+f\left(M\right)\right)+\frac{1-\epsilon}{6}\cdot f\left(M\right)\right)$$
By the definition of $\alpha_1$, it follows that:
$$\alpha_{1}\left(x_{OPT},0,x_{M}\right)\geq\max\left(x_{M},\frac{1-\epsilon}{3}\left(x_{OPT}+x_{M}\right)+\frac{1-\epsilon}{6}\cdot x_{M}\right)$$
The proof that $\alpha_{1}(x_{OPT},x_0,0) \geq  \max(\frac{1-\epsilon}{3}\cdot x_{OPT}+\frac{1-\epsilon}{2}\cdot x_{0},x_{0})$ is similar, 
with the roles of $M$ and $\emptyset$ switched.

\bibliographystyle{plain}
\bibliography{deterministic-submod}

\end{document}